\documentclass[final,5p,times,twocolumn,sort&compress]{elsarticle}
\usepackage{amsthm}
\usepackage{amsfonts} 
\usepackage{amssymb}
\usepackage{amsmath}
\usepackage{stmaryrd}
\SetSymbolFont{stmry}{bold}{U}{stmry}{m}{n}
\usepackage{paralist}
\usepackage{xspace}
\usepackage[table,xcdraw,svgnames]{xcolor}

\allowdisplaybreaks 

\usepackage{pgfplots}
\usepackage{pgfkeys}
\usepackage{tikz-3dplot}
\usepackage{textcomp}
\usepackage{tikz}
\pgfplotsset{compat=\pgfplotsversion}
\usetikzlibrary{backgrounds} 
\usetikzlibrary{positioning,automata,arrows,intersections,calc,shapes,positioning}
\usepgfplotslibrary{patchplots}
\makeatletter
\tikzset{my loop/.style =  {to path={
			\pgfextra{}
			[looseness=12,min distance=10mm]
			\tikz@to@curve@path},font=\sffamily\small
	}}  
	\makeatletter
	
	\tikzstyle{state}=[circle,draw=black,inner sep=0pt,minimum size=10pt]
	\tikzstyle{class red}=[fill=red, regular polygon, regular polygon sides=4]
	\tikzstyle{class yellow}=[fill=yellow, regular polygon, regular polygon sides=3]
	\tikzstyle{class green}=[fill=green, circle]

	\newcommand{\stateVCentered}[1]{\ensuremath{\vcenter{\hbox{\hspace{0.5pt}\ensuremath{#1}}}}}
	
	\newcommand{\statered}{%
		\stateVCentered{%
			\begin{tikzpicture}%
			\node[state, class red, minimum size=9pt]    {};%
			\end{tikzpicture}%
		}%
	}
	
	\newcommand{\stateredInSubscript}{%
		\stateVCentered{%
			\begin{tikzpicture}%
			\node[state, class red, minimum size=5pt]    {};%
			\end{tikzpicture}%
		}
	}
	
	\newcommand{\stateyellow}{%
		\stateVCentered{%
			\begin{tikzpicture}%
			\node[state, class yellow, minimum size=9pt]    {};%
			\end{tikzpicture}%
		}%
	}
	
	\newcommand{\stateyellowInSubscript}{%
		\stateVCentered{%
			\begin{tikzpicture}%
			\node[state, class yellow, minimum size=5pt]    {};%
			\end{tikzpicture}%
		}%
	}
	
	\newcommand{\stategreen}{%
		\stateVCentered{%
			\begin{tikzpicture}%
			\node[state, class green, minimum size=9pt]    {};%
			\end{tikzpicture}%
		}%
	}  
	
	\newcommand{\stategreenInSubscript}{%
		\stateVCentered{%
			\begin{tikzpicture}%
			\node[state, class green, minimum size=5pt]    {};%
			\end{tikzpicture}%
		}%
	}
\newcommand{\mystackrel}[2]{%
 \mathrel{\vbox{\offinterlineskip\ialign{%
  \hfil##\hfil\cr
  $\scriptstyle#1$\cr
  \noalign{\kern.2ex}
  $#2$\cr
}}}}

\newtheorem{theorem}{Theorem}
\newtheorem{proposition}{Proposition}

\newtheorem{lemma}{Lemma}

\newdefinition{definition}{Definition}
\newdefinition{example}{Example}
\newdefinition{remark}{Remark}

\newcommand{\interleavedProd}[2]{#1 \curlywedgedownarrow #2}
\newcommand{\syncProd}[2]{#1 \otimes #2}

\newenvironment{myexample}{\begin{example}}{\hfill$\blacklozenge$\end{example}}
\newenvironment{myproof}[1][]{\begin{proof}[Proof #1]}{\end{proof}}

\renewcommand{\epsilon}{\varepsilon}
\renewcommand{\phi}{\varphi}

\DeclareMathOperator{\proj}{proj}
\newcommand{\vct}[1]{\mathbf{e_{#1}}}
\newcommand{\vctproj}[2][]{\proj_{\vct{#1}}{#2}}
\newcommand{\reals}{\mathbb{R}}

\newcommand{\naturals}{\mathbb{N}}

\newcommand{\setcond}[2]{\lbrace\, #1 \mid #2 \,\rbrace}
\newcommand{\setnocond}[1]{\{#1\}}
\newcommand{\setcardinality}[1]{|#1|}
\newcommand{\interval}[2]{[#1,#2]}
\newcommand{\sd}{\mu} 
\newcommand{\gd}{\rho} 

\newcommand{\distclass}[2][\relord]{[#2]_{#1}}

\newcommand{\dirac}[1]{\delta_{#1}}

\newcommand{\class}[1]{[#1]}

\newcommand{\functionDot}{\,\cdot\,}

\newcommand{\Disc}[1]{\Delta(#1)}

\newcommand{\imdp}[1][I]{\mathfrak{#1}}
\newcommand{\stateSet}{S}
\newcommand{\actionSet}{\mathcal{A}}
\newcommand{\stateActionSet}[1]{\actionSet(#1)}
\newcommand{\startState}[1][s]{\bar{#1}}
\newcommand{\unfold}{\mathtt{UF}}
\newcommand{\fold}{\mathtt{F}}
\newcommand{\foldAction}{f}

\newcommand{\APSet}{\mathtt{AP}}
\newcommand{\APLabelling}{L}

\newcommand{\intTransitionProbability}{\mathit{I}}
\newcommand{\intervalSet}{\mathbb{I}}
\newcommand{\extreme}[1]{\mathtt{{Ext}}({#1})}

\newcommand{\posreals}{\reals_{\geq 0}}

\newcommand{\aut}[1][P]{\mathfrak{#1}}
\newcommand{\transitionRelation}{\mathit{T}}

\newcommand{\tr}{\mathit{tr}}
\newcommand{\source}[1]{\mathit{src}(#1)}
\newcommand{\target}[1]{\mathit{trg}(#1)}

\newcommand{\combined}{\mathrm{c}}
\newcommand{\strongTransition}[2]{{#1 \longrightarrow #2}}
\newcommand{\strongCombinedTransition}[2]{{#1 \longrightarrow_{\combined} #2}}

\newcommand{\relord}[1][\relsymbol]{\mathcal{#1}}
\newcommand{\rel}[1][\relsymbol]{\mathrel{\relord[#1]}}

\newcommand{\idrelord}[1][\idrelsymbol]{\relord[#1]}

\newcommand{\liftrelord}[1][\relord]{\mathcal{L}(#1)}
\newcommand{\liftrel}[1][\relord]{\mathrel{\liftrelord[#1]}}
			
\newcommand{\partitionset}[2][\relord]{#2/#1}
			
\newcommand{\acronym}[1]{\textsl{#1}\xspace}
\newcommand{\PA}{\acronym{PA}}
\newcommand{\PAs}{\acronym{PAs}}

\newcommand{\MDPs}{\acronym{MDPs}}
\newcommand{\MC}{\acronym{MC}}
\newcommand{\MCs}{\acronym{MCs}}			

\newcommand{\bisim}{\sim}

\newcommand{\PAbisim}{\sim^{p}_{aa}}

\newcommand{\IMDP}{\acronym{IMDP}}
\newcommand{\IMDPs}{\acronym{IMDPs}}

\newcommand{\Sch}{\Sigma}
\newcommand{\Env}{\Pi}
\newcommand{\env}{\pi}
\newcommand{\sch}{\sigma}

\newcommand{\uncertainty}[2]{\mathcal{P}^{#1,#2}}
\newcommand{\paruncertainty}[3][\relord]{\mathcal{P}^{#2,#3}_{#1}}

\newcommand{\parcombuncertainty}[2][\relord]{\mathcal{P}^{#2}_{#1}}
\newcommand{\PAparcombuncertainty}[2][\relord]{\mathtt{P}^{#2}_{#1}}

\newcommand{\un}{c}

\newcommand{\transition}[2]{#1 \longrightarrow #2}
\newcommand{\transitionAction}[3]{#1 \mystackrel{#2}{\longrightarrow} #3}

\newcommand{\eqclass}{\mathcal{C}}

\newcommand{\convexComb}[1]{\mathrm{CH}(#1)}
\newcommand{\convexhull}[1]{\convexComb{#1}}

\begin{document} 
	
\begin{frontmatter}
	\title{Compositional Reasoning for Interval Markov Decision Processes}
	\author[mpi,uds]{Vahid~Hashemi}
	\author[uds]{Holger~Hermanns}
	\author[ios]{Andrea~Turrini}
	\address[mpi]{Max Planck Institute for Informatics, Saarbr\"{u}cken, Germany}
	\address[uds]{Department of Computer Science, Saarland University, Saarbr\"{u}cken, Germany}
	\address[ios]{State Key Laboratory of Computer Science, ISCAS, Beijing, China}	

	\begin{abstract}
		Model checking probabilistic CTL properties of Markov decision processes with
		convex uncertainties has been recently investigated by Puggelli et al.
		Such model checking algorithms typically suffer from the state space
		explosion. In this paper, we address probabilistic bisimulation to reduce
		the size of such an MDP while preserving the probabilistic CTL properties it satisfies. %
		In particular, we discuss the key ingredients to build up 
		the operations of parallel composition for composing interval MDP components at run-time.
		More precisely, we investigate how the parallel composition operator for interval MDPs can be defined so as to arrive at a congruence closure.
		As a result, we show that probabilistic bisimulation for interval MDPs is congruence with respect to two facets of parallelism, namely synchronous product and interleaving.
	\end{abstract}
	\begin{keyword}
		Markov Decision Process \sep Interval MDP \sep Compositionality \sep Probabilistic Bisimulation
	\end{keyword}
	\journal{Elsevier}
\end{frontmatter}	

\section{Introduction}
\label{sec:Introduction}
Probability, nondeterminism, and uncertainty are three core aspects
of real systems. \emph{Probability} arises when a system, performing
an action, is able to reach more than one state and we can estimate
the proportion between reaching each of such states: probability can
model both specific system choices (such as flipping a coin, commonly
used in randomized distributed algorithms) and general system
properties (such as message loss probabilities when sending a message
over a wireless medium). \emph{Nondeterminism} represents behaviors
that we can not or we do not want to attach a precise (possibly
probabilistic) outcome to. This might reflect the concurrent execution
of several components at unknown (relative) speeds or behaviors we
keep undetermined for simplifying the system or allowing for different
implementations. \emph{Uncertainty} relates to the fact that not all
system parameters may be known exactly, including exact 
probability values.

Probabilistic automata (\PAs)~\cite{Seg95} extend classical concurrency models in a simple
yet conservative fashion. In probabilistic automata, concurrent
processes may perform probabilistic experiments inside a
transition. 
\PAs are akin to \emph{Markov decision processes} (\MDPs), their fundamental
beauty can be paired with powerful model checking techniques, as
implemented for instance in the PRISM tool~\cite{KNP11}.

In \PAs and \MDPs, probability
values need to be specified precisely. This is often an impediment to
their applicability to real systems. Instead it appears more viable to
specify ranges of probabilities, so as to reflect the uncertainty in
these values. This leads to a model where intervals of probability
values replace probabilities. This is the model studied in this
paper, we call it \emph{interval Markov decision processes}, \IMDPs.

In standard concurrency theory, \emph{bisimulation} plays a central
role as the undisputed reference for distinguishing the behaviour of
systems. Besides for distinguishing systems, bisimulation relations
conceptually allow us to reduce the size of a behaviour representation
without changing its properties (i.e., with respect to logic formulae
the representation satisfies). This is particularly useful to alleviate the state
explosion problem notoriously encountered in model checking. If the
bisimulation is a congruence with respect to a parallel composition
operator used to build up the model out of smaller ones, this can
give rise to a compositional strategy to associate a small model to a
large system without intermediate state space explosion. In several
related settings, this strategy has been proven very
effective~\cite{HK00,CHLS09}. 

Markov chains are known to be closed under interleaving parallelism (if 
considering the continuous-time setting) and under synchronous (also 
called synchronous product) parallelism (if considering the 
discrete-time setting). The more general concept of asynchronous 
parallelism with synchronisation (as in CCS or CSP) is known to require 
nondeterminism so as to arrive at closure properties (yielding \PA for 
discrete time and interactive \MC~\cite{hermanns02} for continuous time).

These observations are conceptually echoed in the setting considered in 
the present paper, albeit for very different reasons. While 
nondeterminism is a genuine asset of \IMDPs, a closure property 
can not be established for asynchronous parallelism with 
synchronisation. It has been recently investigated in~\cite{HashemiHSSTW16}
the possibility of establishing a asynchronous parallelism with 
synchronisation for \IMDP models. However, the underlying construction is problematic since it does not manage 
correctly the spurious distributions. More precisely, for a pair of \IMDP components the equality of the emerged 
sets of spurious distributions as a parallelism result should be guaranteed in order to establish the congruence result.
This fact is not treated precisely in the setting of~\cite{HashemiHSSTW16} for the defined asynchronous parallelism with 
synchronisation. In this work instead \IMDPs are shown to be closed under 
interleaving parallelism, as well as under synchronous parallelism.
This enables us to develop compositionality results with respect to 
bisimulation for these two facets of parallelism.

\paragraph*{Related work}

Compositional specification of uncertain stochastic systems has been explored in various works before. Interval \MCs~\cite{JL91,DBLP:journals/rc/KozineU02} and Abstract \PAs~\cite{DBLP:conf/vmcai/DelahayeKLLPSW11} serve as specification theories for \MCs and
\PAs featuring satisfaction relation, and various refinement relations. In
order to be closed under parallel composition, Abstract \PAs allow general
polynomial constraints on probabilities instead of interval bounds. Since
for Interval \MCs it is not possible to explicitly construct parallel
composition, the problem of whether there is a common implementation of a set
of Interval \MCs is addressed
instead~\cite{DBLP:conf/lata/DelahayeLLPW11}. To the contrary, interval
bounds on \emph{rates} of outgoing transitions work well with parallel
composition in the continuous-time setting of Abstract Interactive \MCs~\cite{DBLP:conf/formats/KatoenKN09}. 
The reason is that unlike
probabilities, rates do not need to sum up to $1$. 
Authors of~\cite{yi1994reasoning} successfully define parallel composition for
interval models by separating synchronizing transitions from the
transitions with uncertain probabilities.

\paragraph*{Organization of the paper} 
We start with necessary preliminaries in
Section~\ref{sec:preliminaries}. In Section~\ref{sec:bisimulation},
we give the definition of probabilistic bisimulation for
\IMDPs and discuss the main results of~\cite{HHK14}. 
Furthermore, we show that the probabilistic bisimulation over \IMDPs is compositional and transitive. 
Finally, in Section~\ref{sec:conclusion} we conclude the paper.

\section{Preliminaries}
\label{sec:preliminaries}

Given $n \in \naturals$, we denote by $\vec{1} \in \reals^{n}$ the unit vector and by $\vec{1}^T$ its transpose. 
In the sequel, the comparison between vectors is element-wise and all vectors are column ones unless otherwise stated. 
For a given set $P \subseteq \reals^{n}$, we denote by $\convexComb{P}$ the convex hull of $P$ and by $\extreme{P}$ the set of extreme points of $P$.
If $P$ is a polytope in $\reals^{n}$ then for each $i \in \setnocond{1, \dotsc, n}$, the projection $\vctproj[i]{P}$ of $P$ is defined as the interval $\interval{\min_{i} P}{\max_{i} P}$ where $\min_{i} P = \min\setcond{x_{i}}{(x_{1},\dotsc,x_{i},\dotsc,x_{n}) \in P}$ and $\max_{i} P = \max\setcond{x_{i}}{(x_{1},\dotsc,x_{i},\dotsc,x_{n}) \in P}$.

We denote by $\intervalSet$ is a set of closed subintervals of $[0,1]$ and, for a given $[a,b] \in \intervalSet$, we let $\inf [a,b] = a$ and $\sup [a,b] = b$.

Given a set $X$, we denote by $\idrelord_{X}$ the identity equivalence relation $\idrelord_{X} = \setcond{(x,x)}{x \in X}$.
We may drop the subscript $X$ from $\idrelord_{X}$ when the set $X$ is clear from the context.

Given two relations $\relord \subseteq X \times Y$ and $\relord[S] \subseteq U \times V$, we denote by $\relord \times \relord[S]$ the relation $\relord \times \relord[S] = \setcond{((x,u),(y,v)) \in (X \times Y) \times (U \times V)}{(x,y) \in \relord, (u,v) \in \relord[S]}$.
If $\relord[X]$ is an equivalence relation on $X$ and $\relord[Y]$ an equivalence relation on $Y$, then $\relord[X] \times \relord[Y]$ is an equivalence relation on $X \times Y$.

For a given set $X$, we denote by $\Disc{X}$ the set of discrete probability distributions over $X$ and by $\dirac{x} \in \Disc{X}$ the \emph{Dirac} distribution on $x$, that is, the distribution such that for each $y \in X$, $\dirac{x}(y) = 1$ if $y = x$, $0$ otherwise. 
Given two sets $X$ and $Y$ and two distributions $\gd_{X} \in \Disc{X}$ and $\gd_{Y} \in \Disc{Y}$, we denote by $\gd_{X} \times \gd_{Y}$ the distribution $\gd_{X} \times \gd_{Y} \in \Disc{X \times Y}$ such that for each $(x,y) \in X \times Y$, $(\gd_{X} \times \gd_{Y})(x,y) = \gd_{X}(x) \cdot \gd_{Y}(y)$.
Given a finite set of indexes $I$, a multiset of distributions $\setcond{\gd_{i} \in \Disc{X}}{i \in I}$, and a multiset of real values $\setcond{p_{i} \in \posreals}{i \in I}$, we say that $\gd$ is the convex combination of $\setcond{\gd_{i} \in \Disc{X}}{i \in I}$ according to $\setcond{p_{i} \in \posreals}{i \in I}$, denoted by $\gd = \sum_{i \in I} p_{i} \cdot \gd_{i}$, if $\sum_{i \in I} p_{i} = 1$ and for each $x \in X$, $\gd(x) = \sum_{i \in I} p_{i} \cdot \gd_{i}(x)$.
For an equivalence relation $\relord$ on $X$ and $\gd_{1}, \gd_{2} \in \Disc{X}$, we write $\gd_{1} \liftrel \gd_{2}$ if for each $\eqclass \in \partitionset{X}$, it holds that $\gd_{1}(\eqclass) = \gd_{2}(\eqclass)$. 
By abuse of notation, we extend $\liftrelord$ to distributions over $\partitionset{X}$, i.e., for $\gd_{1}, \gd_{2} \in \Disc{\partitionset{X}}$, we write $\gd_{1} \liftrel \gd_{2}$ if for each $\eqclass \in \partitionset{X}$, it holds that $\gd_{1}(\eqclass) = \gd_{2}(\eqclass)$.

\subsection{Interval Markov Decision Processes}

Let us formally define Interval Markov Decision Processes. 

\begin{definition}
\label{def:imdp} 
An \emph{Interval Markov Decision Process} (\IMDP) $\imdp$ is a tuple $\imdp = (\stateSet, \startState, \actionSet, \APSet,\APLabelling, \intTransitionProbability)$, where 
$\stateSet$ is a finite set of \emph{states}, 
$\startState \in \stateSet$ is the \emph{initial state}, 
$\actionSet$ is a finite set of \emph{actions}, $\APSet$ is a finite set of \emph{atomic propositions}, 
$\APLabelling \colon \stateSet \to 2^{\APSet}$ is a \emph{labelling function}, 
and 
$\intTransitionProbability \colon \stateSet \times \actionSet \times \stateSet \to \intervalSet$ is an \emph{interval transition probability function} such that for each $s$, there exist $a$ and $s'$ such that $\intTransitionProbability(s,a,s') \neq [0,0]$. We denote by $\class{\imdp}$, the class of all finite-state finite-transition \IMDPs.  
\end{definition} 
We denote by $\stateActionSet{s}$ the set of actions that are enabled from state $s$, i.e., $\stateActionSet{s} = \setcond{a \in \actionSet}{\exists s' \in \stateSet: \intTransitionProbability(s,a,s') \neq [0,0]}$. 
Furthermore, for each state $s$ and action $a \in \stateActionSet{s}$, we let $\transitionAction{s}{a}{\sd_{s}}$ mean that $\sd_{s} \in \Disc{\stateSet}$ is a \emph{feasible distribution}, i.e., for each state $s'$ we have $\sd_{s}(s') \in \intTransitionProbability(s,a,s')$. 
We require that the set $\uncertainty{s}{a} = \setcond{\sd_{s}}{\transitionAction{s}{a}{\sd_{s}}}$ is non-empty for each state $s$ and action $a \in \stateActionSet{s}$.

An \IMDP is initiated in some state $s_{1}$ and then moves in discrete
steps from state to state forming an infinite path $s_{1} \, s_{2} \, s_3
\dots$. One step, say from state $s_{i}$, is performed as follows. First, an
action $a \in \stateActionSet{s}$ is chosen nondeterministically by
\emph{scheduler}. Then, \emph{nature} resolves the uncertainty and chooses
nondeterministically one corresponding feasible distribution $\sd_{s_{i}}
\in \uncertainty{s_{i}}{a}$. Finally, the next state $s_{i+1}$ is chosen
randomly according to the distribution $\sd_{s_{i}}$.
For a more formal treatment of the \IMDP semantics, we refer the reader to~\cite{HashemiHSSTW16,HHK14}.

Observe that the scheduler does not choose an action but a
\emph{distribution} over actions. It is well-known~\cite{Seg95} that such
randomization brings more power in the context of bisimulations. 
Note that for nature this is not the case,
since $\uncertainty{s}{a}$ is closed under convex combinations, thus nature can choose all distributions.

\subsection{Action Agnostic Probabilistic Automata}
\label{sec:PAs}

We now introduce the action agnostic probabilistic automata we use in this paper, based on the probabilistic automata framework \cite{Seg95}, following the notation of \cite{Seg06}.
Note that the probabilistic automata we consider here correspond to the \emph{simple} probabilistic automata of \cite{Seg95}.
In practice, we consider the subclass of (simple) probabilistic automata of \cite{Seg95} having as set of actions the same singleton $\setnocond{\foldAction}$, that is, all transitions are labelled by the same external action $\foldAction$.
Since this action is unique, we just drop it from the definitions.

\begin{definition}
\label{def:pa}
	An \emph{(action agnostic) probabilistic automaton} (\PA) is a tuple $\aut = (\stateSet, \startState, \APSet, \APLabelling, \transitionRelation)$, where $\stateSet$ is a set of \emph{states}, $\startState \in \stateSet$ is the \emph{start} state, $\APSet$ is a finite set of \emph{atomic propositions}, $\APLabelling \colon \stateSet \to 2^{\APSet}$ is a \emph{labelling function}, and $\transitionRelation \subseteq \stateSet \times \Disc{\stateSet}$ is a \emph{probabilistic transition relation}.
\end{definition}
We denote by $\class{\aut}$ the class of all finite-state finite-transition probabilistic automata and we assume that each state in $\stateSet$ is reachable from $\startState$.  
We may drop \emph{action agnostic} since this is the only type of probabilistic automata we consider. 
The start state is also called the \emph{initial} state; 
we let $s$, $t$, $u$, $v$, and their variants with indices range over $\stateSet$.

We denote the generic elements of a probabilistic automaton $\aut$ by $\stateSet$, $\startState$, $\APSet$, $\APLabelling$, $\transitionRelation$, and we propagate primes and indices when necessary. 
Thus, for example, the probabilistic automaton $\aut'_{i}$ has states $\stateSet'_{i}$, start state $\startState'_{i}$, and transition relation $\transitionRelation'_{i}$. 

A transition $\tr = (s, \sd) \in \transitionRelation$, also written $\strongTransition{s}{\sd}$, is said to \emph{leave} from state $s$ and to \emph{lead} to the measure $\sd$.
We denote by $\source{\tr}$ the \emph{source} state $s$ and by $\target{\tr}$ the \emph{target} measure $\sd$, also denoted by $\sd_{\tr}$.
We also say that $s$ enables the transition $(s,\sd)$ and that $(s,\sd)$ is enabled from $s$.

\begin{figure}
	\centering
	\begin{tikzpicture}[->,>=stealth',auto, scale=0.9]
	\scriptsize
	\path[use as bounding box] (-0.25,-1.25) rectangle (4.75,1.35);
	\tikzstyle{every state}=[fill=none,draw=black,text=black,shape=circle]
	
	\node[state] (s) at (0,0) {$\vphantom{lg}\startState$};
	\node[state] (r) at (2.75,1) {$\vphantom{lg}r$};
	\node[state] (y) at (2.75,0) {$\vphantom{lg}y$};
	\node[state] (g) at (2.75,-1) {$\vphantom{lg}g$};
	\node[state,class red] (red) at (4.5,1) {};
	\node[state,class yellow] (yellow) at (4.5,0) {};
	\node[state,class green] (green) at (4.5,-1) {};
	
	\draw[name path=t1] (s.east) to node [below, very near end] {$0.3$} (r);
	\draw[name path=t2] (s.east) to node [above, very near end] {$0.1$} (y);
	\draw[name path=t3] (s.east) to node [above, very near end] {$0.6$} (g);
	
	\draw (r.west) to [bend right] node [above, very near end] {$1$} (s);
	
	\draw (r.east) to node [above, very near end] {$1$} (red);
	\draw (y.east) to node [above, very near end] {$1$} (yellow);
	\draw (g.east) to node [above, very near end] {$1$} (green);
	
	\draw (g.west) to [bend left] node [below, very near end] {$1$} (s);
	
	\path[name path=c] (s) circle (1.25);
	\draw[name intersections={of=t3 and c, by=i3},
	name intersections={of=t1 and c, by=i1},
	-,shorten >=0pt]
	(i3) to[bend right] (i1);
	\end{tikzpicture}
	\caption{An example of \PAs: the \PA $\aut[E]$}
	\label{fig:examplePA}
\end{figure}
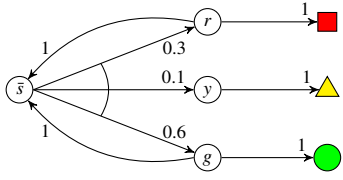

\begin{myexample}
\label{ex:pa}
An example of \PA is the one shown in Figure~\ref{fig:examplePA}:
the set of states is $\stateSet = \setnocond{\startState, r, y, g, \statered, \stateyellow, \stategreen}$, the start state is $\startState$, the set of atomic propositions is $\APSet = \stateSet$, the labelling function $\APLabelling$ is such that for each $s \in \stateSet$, $\APLabelling(s) = s$, and the transition relation $\transitionRelation$ contains the following transitions: 
$\strongTransition{\startState}{\gd}$ with $\gd = \setnocond{(r, 0.3), (y, 0.1), (g, 0.6)}$, 
$\strongTransition{r}{\dirac{\stateredInSubscript}}$, 
$\strongTransition{y}{\dirac{\stateyellowInSubscript}}$, 
$\strongTransition{g}{\dirac{\stategreenInSubscript}}$, 
$\strongTransition{r}{\dirac{\startState}}$,
and 
$\strongTransition{g}{\dirac{\startState}}$.
\end{myexample}

\subsubsection{Synchronous Product}

The following definition of synchronous product is a variation of the definition of parallel composition provided in \cite{Seg95,Seg06}, where the synchronization occurs for each pair of enabled transitions.
This corresponds to the original definition of parallel composition for probabilistic automata having all transitions labelled by the same external action.
\begin{definition}
	\label{def:PAparallelComposition}
	Given two \PAs $\aut_{1}$ and $\aut_{2}$, the \emph{synchronous product} of $\aut_{1}$ and $\aut_{2}$, denoted by $\syncProd{\aut_{1}}{\aut_{2}}$, is the probabilistic automaton $\aut = (\stateSet, \startState, \APSet, \APLabelling, \transitionRelation)$ where 
		$\stateSet = \stateSet_{1} \times \stateSet_{2}$;
		$\startState = (\startState_{1}, \startState_{2})$;
		$\APSet = \APSet_{1} \cup \APSet_{2}$;
		for each $(s_{1},s_{2}) \in \stateSet$, $\APLabelling(s_{1},s_{2}) = \APLabelling_{1}(s_{1}) \cup \APLabelling_{2}(s_{2})$; 
		and
		$\transitionRelation = \setcond{((s_{1},s_{2}), \sd_{1} \times \sd_{2})}{\text{$(s_{1}, \sd_{1}) \in \transitionRelation_{1}$ and $(s_{2}, \sd_{2}) \in \transitionRelation_{2}$}}$.
\end{definition}
For two \PAs $\aut_{1}$ and $\aut_{2}$ and their synchronous product $\syncProd{\aut_{1}}{\aut_{2}}$, we refer to $\aut_{1}$ and $\aut_{2}$ as the component automata and to $\syncProd{\aut_{1}}{\aut_{2}}$ as the product automaton.

\subsubsection{Probabilistic Bisimulation}
As for the definition of synchronous product, the following definition of (strong) probabilistic bisimulation is a variation of the definition provided in \cite{Seg06}, where all actions are treated as being the same external action.
We first introduce the definition of combined transition.

\begin{definition}
\label{def:PAcombinedTransition}
	Given a \PA $\aut$ and a state $s$, we say that $s$ enables a combined transition reaching the distribution $\sd$, denoted by $\strongCombinedTransition{s}{\sd}$, if there exist a finite set of indexes $I$, a multiset of transitions $\setcond{(s, \sd_{i}) \in \transitionRelation}{i \in I}$, and a multiset of real values $\setcond{p_{i} \in \posreals}{i \in I}$ such that 
		$\sum_{i \in I} p_{i} = 1$ and
		$\sd = \sum_{i \in I} p_{i} \cdot \sd_{i}$.
\end{definition}

\begin{definition}
\label{def:PAbisim}
	Given a \PA $\aut$, an equivalence relation $\relord \subseteq \stateSet \times \stateSet$ is a \emph{(strong) (action agnostic) probabilistic bisimulation} on $\aut$ if, for each $(s,t) \in \relord$, $\APLabelling(s) = \APLabelling(t)$ and for each $\strongTransition{s}{\sd_{s}}$, there exists a combined transition $\strongCombinedTransition{t}{\sd_{t}}$ such that $\sd_{s} \liftrel \sd_{t}$.
	
	Given two states $s$ and $t$, we say that $s$ and $t$ are probabilistically bisimilar, denoted by $s \PAbisim t$, if there exists a probabilistic bisimulation $\relord$ on $\aut$ such that $(s,t) \in \relord$.
	
	Given two \PAs $\aut_{1}$ and $\aut_{2}$, we say that $\aut_{1}$ and $\aut_{2}$ are probabilistically bisimilar, denoted by $\aut_{1} \PAbisim \aut_{2}$, if there exists a probabilistic bisimulation $\relord$ on the disjoint union of $\aut_{1}$ and $\aut_{2}$ such that $(\startState_{1},\startState_{2}) \in \relord$.
\end{definition}

\begin{proposition}
\label{pro:PAbisimSyncProduct}
	Given three \PAs $\aut_{1}$, $\aut_{2}$, and $\aut_{3}$, if $\aut_{1} \PAbisim \aut_{2}$, then $\syncProd{\aut_{1}}{\aut_{3}} \PAbisim \syncProd{\aut_{2}}{\aut_{3}}$.
\end{proposition}
\begin{proof}
The proof is a minor adaptation of the corresponding proof (cf.~\cite{Seg95}) for the original definition of probabilistic bisimulation and parallel composition of \PAs.

In the following, we use the subscript ``$j,3$'' with $j \in \setnocond{1,2}$ to refer to the component of the \PA $\aut_{j,3} = \syncProd{\aut_{j}}{\aut_{3}}$.

Let $\relord$ be the probabilistic bisimulation justifying $\aut_{1} \PAbisim \aut_{2}$ and $\relord' = \relord \times \idrelord_{\stateSet_{3}}$;
we claim that $\relord'$ is a probabilistic bisimulation between $\syncProd{\aut_{1}}{\aut_{3}}$ and $\syncProd{\aut_{2}}{\aut_{3}}$.
The fact that $\relord'$ is an equivalence relation follows trivially by its definition and the fact that $\relord$ is an equivalence relation.
The fact that $((\startState_{1},\startState_{3}), (\startState_{2},\startState_{3}))$ follows immediately by the hypothesis that $(\startState_{1}, \startState_{2}) \in \relord$ and $(\startState_{3}, \startState_{3}) \in \idrelord_{\stateSet_{3}}$.

Let $((s_{1},s_{3}),(s_{2},s_{3})) \in \relord'$.
Assume, without loss of generality, that $s_{1} \in \stateSet_{1}$ and $s_{2} \in \stateSet_{2}$;
the other cases are similar.
The fact that $\APLabelling_{1,3}(s_{1},s_{3}) = \APLabelling_{2,3}(s_{2},s_{3})$ is straightforward, since by definition of synchronous product and the hypothesis that $s_{1} \rel s_{2}$, we have that $\APLabelling_{1,3}(s_{1},s_{3}) = \APLabelling_{1}(s_{1}) \cup \APLabelling_{3}(s_{3}) = \APLabelling_{2}(s_{2}) \cup \APLabelling_{3}(s_{3}) = \APLabelling_{2,3}(s_{2},s_{3})$, as required.

Consider now a transition $\strongTransition{(s_{1},s_{3})}{\sd_{1,3}}$.
By definition of synchronous product, there exist $\sd_{1}$ and $\sd_{3}$ such that $\strongTransition{s_{1}}{\sd_{1}} \in \transitionRelation_{1}$, $\strongTransition{s_{3}}{\sd_{3}} \in \transitionRelation_{3}$, and $\sd_{1,3} = \sd_{1} \times \sd_{3}$.
Since $s_{1} \rel s_{2}$, it follows that there exists a combined transition $\strongCombinedTransition{s_{2}}{\sd_{2}}$ such that $\sd_{1} \liftrel \sd_{2}$.
Let $I$ be a the finite set of indexes, $\setcond{(s_{2}, \sd_{2,i}) \in \transitionRelation_{2}}{i \in I}$ be a multiset of transitions, and $\setcond{p_{i} \in \posreals}{i \in I}$ be a multiset of real values such that $\sum_{i \in I} p_{i} = 1$ and $\sd_{2} = \sum_{i \in I} p_{i} \cdot \sd_{2,i}$.
By definition of synchronous product, it follows that for each $i \in I$, $\strongTransition{(s_{2},s_{3})}{\sd_{2,i} \times \sd_{3}} \in \transitionRelation_{2,3}$, hence we have the combined transition $\strongCombinedTransition{(s_{2},s_{3})}{\sd_{2} \times \sd_{3}}$.
By standard properties of lifting (see, e.g.,~\cite{TurriniH14}), it follows that $\sd_{1} \times \sd_{3} \liftrel[\relord'] \sd_{2} \times \sd_{3}$, as required.
\end{proof}

\subsection{\IMDPs vs. \PAs} 
\label{sec:unfoldIMDPs}
A cornerstone towards establishing compositional reasoning for \IMDPs essentially relies on \emph{transformations} from \IMDPs to \PAs and 
vice versa. To this aim, we define two mappings namely, \emph{unfolding} which unfolds a given \IMDP as a \PA and \emph{folding}
which transforms a given \PA to an \IMDP. Formally,

\begin{definition}[Unfolding mapping]\label{def:unfoldMap}	
An unfolding mapping $\unfold \colon \class{\imdp} \to \class{\aut}$ is a function that maps a given \IMDP $\imdp = (\stateSet, \startState, \actionSet, \APSet,\APLabelling, \intTransitionProbability)$ to the \PA $\aut = (\stateSet, \startState, \APSet, \APLabelling, \transitionRelation)$ where $\transitionRelation = \setcond{(s, \sd)}{s \in \stateSet, \exists a \in \stateActionSet{s}:\sd \in \extreme{\uncertainty{s}{a}}}$.
\end{definition}

\begin{figure}[!tbp]
	\centering
		\begin{tikzpicture}[->,>=stealth',auto,
			state/.style={draw,circle,minimum size=15pt,inner sep=0},
			trans/.style={font=\footnotesize},
			prob/.style={font=\scriptsize},
			junct/.style={circle,fill,minimum size=3pt,inner sep=0,outer sep=0}
		,scale=0.9]
		
			\node[state] (IMDPt) at (0,0) {$t$};
			\node[state] (IMDPu) at ($(IMDPt)+(-1,-2)$) {$u$};
			\node[state] (IMDPv) at ($(IMDPt)+(1,-2)$) {$v$};
			\node[junct] (IMDPtr) at ($3/5*(IMDPt)+1/5*(IMDPu)+1/5*(IMDPv)$) {};
			\node[anchor=center] at ($(IMDPt)+(0,-3)$) {(a)};
			
			\draw[-] (IMDPt) to node[trans,left] {$a$} (IMDPtr);
			\draw (IMDPtr) to node[prob,left] {$[0.1,0.3]$} (IMDPu);
			\draw (IMDPtr) to node[prob,right] {$[0.8,1]$} (IMDPv);

			\node[state] (PAt) at (4,0) {$t$};
			\node[state] (PAu) at ($(PAt)+(-1,-2)$) {$u$};
			\node[state] (PAv) at ($(PAt)+(1,-2)$) {$v$};
			\node[junct] (PAtr1) at ($3/5*(PAt)+2/5*(PAu)$) {};
			\node[junct] (PAtr2) at ($3/5*(PAt)+2/5*(PAv)$) {};
			\node[anchor=center] at ($(PAt)+(0,-3)$) {(b)};
			
			\draw[-] (PAt) to (PAtr1);
			\draw (PAtr1) to node[prob,left,near end] {$0.1$} (PAu);
			\draw (PAtr1) to node[prob,left,very near end] {$0.9$} (PAv);

			\draw[-] (PAt) to (PAtr2);
			\draw (PAtr2) to node[prob,right,very near end] {$0.2$} (PAu);
			\draw (PAtr2) to node[prob,right,near end] {$0.8$} (PAv);

			\draw[-latex', line width=1pt] ($1/2*(IMDPt) + 1/2*(IMDPv) + (1,0)$) to node[above] {$\unfold$} ($1/2*(PAt) + 1/2*(PAu) - (1,0)$);
		\end{tikzpicture}
	\caption{Unfolding \IMDP $\imdp$ to \PA $\aut$}
	\label{fig:unfoldIMDP}
\end{figure}
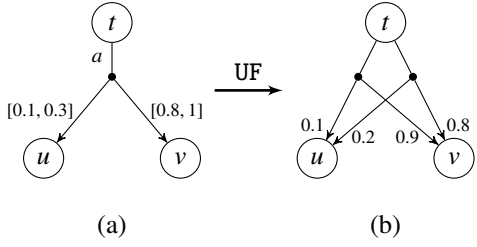
It is worthy to note that the unfolding mapping might transform an \IMDP to a \PA with an exponentially larger size. 
This is in fact due to the exponential blow up in the number of transitions in the resultant \PA which in turn depends on the number of extreme points of the polytope constructed for each state and action in the given \IMDP. 
An example of unfolding is given in Figure~\ref{fig:unfoldIMDP}.

In order to transform a given \PA to an instance of \IMDPs, we use the folding mapping defined as follows:
\begin{definition}[Folding mapping]
\label{def:FoldMap}	
The folding mapping $\fold \colon \class{\aut} \to \class{\imdp}$ transforms a \PA $\aut = (\stateSet, \startState, \APSet, \APLabelling, \transitionRelation)$ to the \IMDP $\imdp = (\stateSet, \startState, \setnocond{\foldAction}, \APSet, \APLabelling, \intTransitionProbability)$ where, for each $s, t \in \stateSet$, $\intTransitionProbability(s, \foldAction, t) = \vctproj[t]{\convexhull{\setcond{\sd}{(s, \sd) \in \transitionRelation}}}$, where each component $\vct{uv}$ of the vector $\vct{u} \in \reals^{\setcardinality{\stateSet}}$ is defined as $\vct{uv}=\delta_{u}(v)$.
\end{definition}

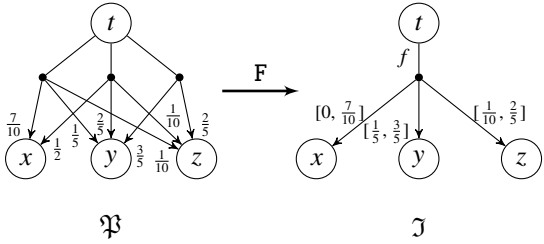
\begin{figure}[!tbp]
	\centering
		\begin{tikzpicture}[->,>=stealth',auto,
			state/.style={draw,circle,minimum size=15pt,inner sep=0},
			trans/.style={font=\footnotesize},
			prob/.style={font=\scriptsize},
			junct/.style={circle,fill,minimum size=3pt,inner sep=0,outer sep=0}
		,scale=0.9]
		
			\node[state] (PAt) at (0,0) {$t$};
			\node[state] (PAy) at ($(PAt)+(0,-2)$) {$y$};
			\node[state] (PAx) at ($(PAy)-(1.25,0)$) {$x$};
			\node[state] (PAz) at ($(PAy)+(1.25,0)$) {$z$};
			\node[junct] (PAtr1) at ($3/5*(PAt)+2/5*(PAx)-(0.5,0)$) {};
			\node[junct] (PAtr2) at ($3/5*(PAt)+2/5*(PAz)+(0.5,0)$) {};
			\node[junct] (PAtr3) at ($3/5*(PAt)+2/5*(PAy)$) {};
			\node[anchor=center] at ($(PAt)+(0,-3)$) {$\aut$};
			
			\draw[-] (PAt) to (PAtr1);
			\draw (PAtr1) to node[prob,left,near end] {$\frac{7}{10}$} (PAx);
			\draw (PAtr1) to node[prob,left,very near end] {$\frac{1}{5}$} (PAy);
			\draw (PAtr1) to node[prob,below,very near end] {$\frac{1}{10}$} (PAz);

			\draw[-] (PAt) to (PAtr3);
			\draw (PAtr3) to node[prob,below,near end] {$\frac{1}{2}$} (PAx);
			\draw (PAtr3) to node[prob,left,near end] {$\frac{2}{5}$\kern-2pt} (PAy);
			\draw (PAtr3) to node[prob,above,very near end] {$\frac{1}{10}$} (PAz);

			\draw[-] (PAt) to (PAtr2);
			\draw (PAtr2) to node[prob,below,very near end] {~~~~$\frac{3}{5}$} (PAy);
			\draw (PAtr2) to node[prob,right,near end] {$\frac{2}{5}$} (PAz);

			\node[state] (IMDPt) at (4.5,0) {$t$};
			\node[state] (IMDPy) at ($(IMDPt)+(0,-2)$) {$y$};
			\node[state] (IMDPx) at ($(IMDPy)-(1.5,0)$) {$x$};
			\node[state] (IMDPz) at ($(IMDPy)+(1.5,0)$) {$z$};
			\node[junct] (IMDPtr) at ($3/5*(IMDPt)+2/5*(IMDPy)$) {};
			\node[anchor=center] at ($(IMDPt)+(0,-3)$) {$\imdp$};
			
			\draw[-] (IMDPt) to node[trans,left] {$\foldAction$} (IMDPtr);
			\draw (IMDPtr) to node[prob,left] {$[0,\frac{7}{10}]$} (IMDPx);
			\draw (IMDPtr) to node[prob,left,very near end] {$[\frac{1}{5},\frac{3}{5}]$} (IMDPy);
			\draw (IMDPtr) to node[prob,right] {$[\frac{1}{10},\frac{2}{5}]$} (IMDPz);

			\draw[-latex', line width=1pt] ($1/2*(PAt) + 1/2*(PAz) + (1,0)$) to node[above] {$\fold$} ($1/2*(IMDPt) + 1/2*(IMDPx) - (1,0)$);
		\end{tikzpicture}
	\caption{Folding a \PA $\aut$ as an \IMDP $\imdp$}
	\label{fig:foldMap}
\end{figure}

An example of the folding mapping is shown in Figure~\ref{fig:foldMap}. 
The \PA $\aut$ has three transitions from $t$ with label $a$; 
in particular, it is worthwhile to note that for all these transitions the probability of reaching $y$ is larger than the probability of reaching $z$, so this has to happen for every combined transition leaving $t$.
According to Def.~\ref{def:FoldMap}, the folding of $\aut$ is the \IMDP $\imdp$. 
It is immediate to see that the unfolding mapping is not surjective as there may be some probabilistic transitions in the generated \IMDP specification which cannot be mapped to a probability distribution in the given \PA. 
In fact, one of such distributions is $\sd_{o}$ such that $\sd_{o}(x) = \frac{2}{5}$, $\sd_{o}(y) = \frac{1}{5}$, and $\sd_{o}(z) = \frac{2}{5}$ that clearly violates the condition $\sd_{o}(y) > \sd_{o}(z)$.
This is better recognizable by comparing the corresponding polytopes in a graphical way.
\begin{figure}[!tbp]
	\centering
\includegraphics[scale=1]{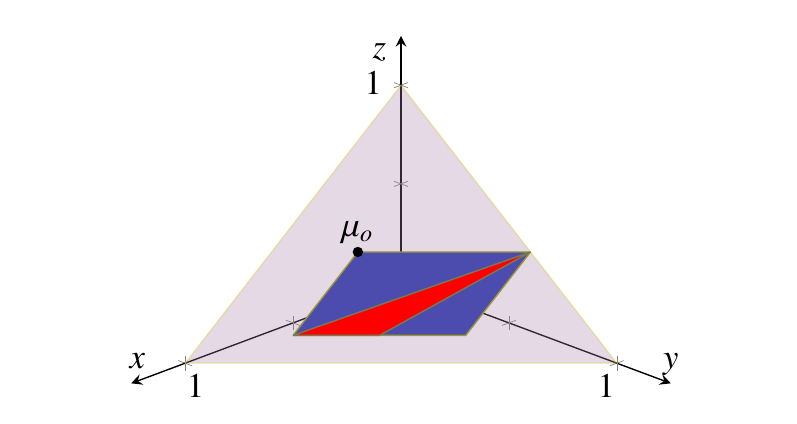}
	\caption{Comparison of polytopes resulted from folding mapping $\fold$}
	\label{fig:polytopes}	
\end{figure}
Figure~\ref{fig:polytopes} shows the three polytopes involved in $\imdp$:
the purplish large triangular polytope is the standard \emph{2-simplex} in the three dimensional space;
the reddish small triangular and the bluish parallelogram-like polytopes represent the convex hull of $\left\{\left(\frac{7}{10},\frac{1}{5},\frac{1}{10}\right), \left(\frac{1}{2},\frac{2}{5},\frac{1}{10}\right), \left(0,\frac{3}{5},\frac{2}{5}\right)\right\}$ and the polytope $\uncertainty{t}{\foldAction}$, respectively, both being a sub-polytope of the 2-simplex. 
Clearly there are points in $\uncertainty{t}{\foldAction}$ that do not belong to the reddish polytope, such as the black dot corresponding to $\sd_{o}$.
	
\begin{lemma}
	Given the folding and unfolding mappings $\fold \colon \class{\aut} \to \class{\imdp}$ and $\unfold \colon \class{\imdp} \to \class{\aut}$, in general:  
	\begin{enumerate}
		\item $\unfold({\fold(\aut)})\neq \aut$
		\item $\fold({\unfold(\imdp)})\neq \imdp$
	\end{enumerate}
\end{lemma}
As we will discuss later, the general incompleteness property of the folding mapping does not influence on the generality of our compositional reasoning for \IMDP specifications. 
We will dive into this point later in Section~\ref{Sec:compositionality}.

\section{Probabilistic Bisimulation for Interval MDPs}
\label{sec:bisimulation} 

We now recall the main results on probabilistic bisimulation
for \IMDPs, as developed in~\cite{HHK14}. In this work, we consider the notion of probabilistic bisimulation for the
cooperative resolution of nondeterminism. This semantics is very natural
in the context of verification of parallel systems with uncertain
transition probabilities in which we assume that scheduler and nature are
resolved \emph{cooperatively} in the most \emph{adversarial} way.
Moreover, resolution of a feasible probability distribution respecting the
interval constraints can be either done \emph{statically}~\cite{JL91},
i.e., at the beginning once for all, or \emph{dynamically}~\cite{SVA06,
Iyengar05}, i.e., independently for each computation step. In this paper, we focus
on dynamic approach in resolving the stochastic nondeterminism that is 
easier to work with algorithmically and can be seen as a relaxation of the static 
approach that is often intractable~\cite{DBLP:conf/tacas/BenediktLW13, 
DBLP:conf/fossacs/ChatterjeeSH08}. 

Let $\transition{s}{\sd_{s}}$ denote that a transition from $s$ to $\sd_{s}$ can be taken cooperatively, i.e., that there is a scheduler $\sch \in \Sch$ and a nature $\env \in \Env$ such that $\sd_{s} = \sum_{a \in \stateActionSet{s}} \sch(s)(a) \cdot \env(s,a)$. 
In other words, $\transition{s}{\sd_{s}}$ if $\sd_{s} \in \convexComb{\bigcup_{a \in \stateActionSet{s}} \uncertainty{s}{a}}$.
\begin{definition}[cf.~\cite{HHK14}]
\label{def:bisim-cooperative} 
	Given an \IMDP $\imdp$, let $\relord \subseteq \stateSet \times \stateSet$ be an equivalence relation. 
	We say that $\relord$ is a \emph{probabilistic bisimulation} if for each $(s,t) \in \relord$ we have that $\APLabelling(s) = \APLabelling(t)$ and for each $\transition{s}{\sd_{s}}$ there exists $\transition{t}{\sd_{t}}$ such that $\sd_{s} \liftrel \sd_{t}$.
	Furthermore, we write $s \bisim_{\un} t$ if there is a probabilistic bisimulation $\relord$ such that $(s,t) \in \relord$. 
\end{definition}
Intuitively, each (cooperative) step of scheduler and nature from
state $s$ needs to be matched by a (cooperative) step of scheduler and
nature from state $t$; symmetrically, $s$ also needs to match $t$. 
In order to support the compositional reasoning, $\bisim_{\un}$ needs to be an equivalence relation. It is not difficult to see that $\bisim_{\un}$ is reflexive and symmetric. 
What remains is to show that it is also transitive. This is indeed a property of $\bisim_{\un}$, as stated by the following proposition:
\begin{theorem}
	\label{thm:IMDPBisimTransitivity}
	Given three \IMDPs $\imdp_{1}$, $\imdp_{2}$, and $\imdp_{3}$, if $\imdp_{1} \bisim_{\un} \imdp_{2}$ and $\imdp_{2} \bisim_{\un} \imdp_{3}$, then $\imdp_{1} \bisim_{\un} \imdp_{3}$. 
\end{theorem}
\begin{proof}
	Let $\relord_{12}$ and $\relord_{23}$ be the equivalence relations underlying $\imdp_{1} \bisim_{\un} \imdp_{2}$ and $\imdp_{2} \bisim_{\un} \imdp_{3}$, respectively. 
	Let $\relord_{13}$ be the symmetric and transitive closure of the set $\setcond{(s_{1},s_{3})}{\exists s_{2}. s_{1} \rel_{12} s_{2} \wedge s_{2} \rel_{23} s_{3}} \cup \setcond{(s_{1},s'_{1}) \in \stateSet_{1} \times \stateSet_{1}}{(s_{1},s'_{1}) \in \relord_{12}} \cup \setcond{(s_{3},s'_{3}) \in \stateSet_{3} \times \stateSet_{3}}{(s_{3},s'_{3}) \in \relord_{23}}$.
	We claim that $\relord_{13}$ is a probabilistic bisimulation justifying $\imdp_{1} \bisim_{\un} \imdp_{3}$.
	
	The fact that $\startState_{1} \rel_{13} \startState_{3}$ is trivial since by hypothesis we have that $\startState_{1} \rel_{12} \startState_{2}$ and $\startState_{2} \rel_{23} \startState_{3}$, so $(\startState_{1}, \startState_{3}) \in \relord_{13}$ by construction.
	
	In the following, assume that $s_{1} \in \stateSet_{1}$ and $s_{3} \in \stateSet_{3}$; 
	the other cases are similar.
	
	The labelling is respected: 
	for each $s_{1} \rel_{13} s_{3}$, we have that there exists $s_{2}$ such that $s_{1} \rel_{12} s_{2}$ and $s_{2} \rel_{23} s_{3}$; 
	this implies that $\APLabelling_{1}(s_{1}) = \APLabelling_{2}(s_{2})$ and $\APLabelling_{2}(s_{2}) = \APLabelling_{3}(s_{3})$, thus $\APLabelling_{1}(s_{1}) = \APLabelling_{3}(s_{3})$ as required.
	
	To complete the proof, consider $s_{1} \rel_{13} s_{3}$ and $\transition{s_{1}}{\sd_{1}}$.
	By hypothesis, there exists $s_{2}$ such that $s_{1} \rel_{12} s_{2}$ and $s_{2} \rel_{23} s_{3}$;
	moreover, by $\relord_{12}$ being a probabilistic bisimulation, we know that there exists $\transition{s_{2}}{\sd_{2}}$ such that $\sd_{1} \liftrel[\relord_{12}] \sd_{2}$.
	Since $\relord_{23}$ is a probabilistic bisimulation, we have that there exists $\transition{s_{3}}{\sd_{3}}$ such that $\sd_{2} \liftrel[\relord_{23}] \sd_{3}$.
	By construction of $\relord_{13}$ and the properties of lifting, it follows that $\sd_{1} \liftrel[\relord_{13}] \sd_{3}$, as required.
\end{proof}

It is shown in~\cite{HHK14} that $\bisim_\un$ is sound with respect to the PCTL properties.
Furthermore, probabilistic bisimulation for \IMDPs is computed using 
standard partition refinement approach~\cite{KS90,PTarjan87} in which the 
core part is to verify the violation of bisimulation definition that can in turn be 
done by checking the inclusion of polytopes defined as follows.   
For $s  \in \stateSet$ and an action $a \in \actionSet$, recall that
$\uncertainty{s}{a}$ denotes the polytope of feasible successor
distributions over \emph{states} with respect to taking the action $a$ in
the state $s$. %
By $\paruncertainty{s}{a}$, we denote the polytope of feasible
successor distributions over \emph{equivalence classes} of $\rel$ with
respect to taking the action $a$ in the state $s$. 
Formally, for $\sd \in \Disc{\partitionset{\stateSet}}$ we set $\sd \in \paruncertainty{s}{a}$ if, for each $\eqclass \in \partitionset{\stateSet}$, it is 
\[
	\sd(\eqclass) \in \left [ {\sum_{s' \in \eqclass}{\inf \intTransitionProbability(s, a, s')}}, {\sum_{s' \in \eqclass}{\sup \intTransitionProbability(s, a, s')}} \right ]\text{.}
\]
Furthermore, we define $\parcombuncertainty{s} = \convexComb{\bigcup_{a \in \stateActionSet{s}} \paruncertainty{s}{a}}$, the set of feasible successor
distributions over $\partitionset{\stateSet}$ with respect to taking
an \emph{arbitrary} distribution over enabled actions in state $s$. %
As specified in~\cite{HHK14}, checking violation of a given pair of states
amounts to check equality of the corresponding constructed polytopes for
the states. 

\section{Compositional Reasoning for \IMDPs}
\label{Sec:compositionality}
The compositional reasoning is a widely used technique (see, e.g.,~\cite{CGMTZ96,HK00,KKZJ07}) that permits to deal with large systems.
In particular, a large system is decomposed into multiple components running in parallel;
such components are then minimized by replacing each of them by a bisimilar but smaller one so that the overall behaviour remains unchanged.
In order to apply this technique, bisimulation has first to be extended to pairs of components and then to be shown to be transitive and preserved by the synchronous product operator.
The extension to a pair of components is trivial and commonly done (see, e.g.,~\cite{CS02,Seg95}):
\begin{definition}
\label{def:bisimPairIMDPs}
	Given two \IMDPs $\imdp_{1}$ and $\imdp_{2}$, we say that they are \emph{probabilistic bisimilar}, denoted by $\imdp_{1} \bisim_{\un} \imdp_{2}$, if there exists a probabilistic bisimulation on the disjoint union of $\imdp_{1}$ and $\imdp_{2}$ such that $\startState_{1} \bisim_{\un} \startState_{2}$.
\end{definition}

The next step is to define the synchronous product for \IMDPs:

\begin{definition}
\label{def:IMDPparallelComposition}
	Given two \IMDPs $\imdp_{1}$ and $\imdp_{2}$, we define the \emph{synchronous product} of $\imdp_{1}$ and $\imdp_{2}$ as 
	\[\syncProd{\imdp_{1}}{\imdp_{2}}:= \fold(\syncProd{\unfold(\imdp_{1})}{\unfold(\imdp_{2})})\text{.}\] 
\end{definition}
A schematic representation of constructing the synchronous product of two \IMDPs $\imdp_{1}$ and $\imdp_{2}$ is given in Figure~\ref{fig:schematic}. 
As discussed earlier, the folding mapping from \PA to \IMDP, i.e. the red arrow, is not complete and in principle, this transformation may add additional behavior to the resultant system.
For each state and action in the resultant \IMDP, these extra behaviors are essentially a set of probability distributions that do not belong to the convex hull of the enabled probability distributions for that state in the original \PA. 
At first sight, these extra behaviors generated from the folding mapping might be seen as an impediment towards showing that $\bisim_{\un}$ is a congruence for the synchronous product. 
Fortunately, as it is shown by the next theorem, these extra probability distributions are in fact \emph{spurious} and do not affect the congruence result.   
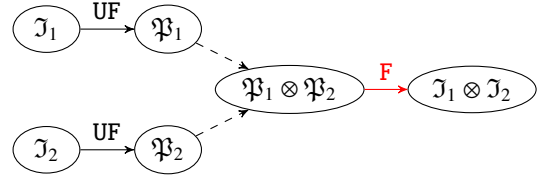
\begin{figure}[tbp]
\centering
	\begin{tikzpicture}[
			->,>=stealth',auto,
			model/.style={draw,ellipse,inner sep=2pt, minimum width=25pt}
	,scale=0.8]
			
		\node[model] (imdp1) at (0,0) {$\imdp_{2}$};
		\node[model] (imdp2) at ($(imdp1) + (0,2)$) {$\imdp_{1}$};
		\node[model] (pa1) at ($(imdp1) + (2,0)$) {$\aut_{2}$};
		\node[model] (pa2) at ($(imdp2) + (2,0)$) {$\aut_{1}$};
		\node[model] (pa12) at ($0.5*(pa1) + 0.5*(pa2) + (2,0)$) {$\syncProd{\aut_{1}}{\aut_{2}}$};
		\node[model] (imdp12) at ($(pa12) + (3,0)$) {$\syncProd{\imdp_{1}}{\imdp_{2}}$};
		
		\draw (imdp1) to node[above] {$\unfold$} (pa1);
		\draw (imdp2) to node[above] {$\unfold$} (pa2);
		\draw[dashed] (pa1) to (pa12);
		\draw[dashed] (pa2) to (pa12);
		\draw[red] (pa12) to node[above] {$\fold$} (imdp12);
	\end{tikzpicture}	
	\caption{Schematic representation of the synchronous product of \IMDPs $\imdp_{1}$ and $\imdp_{2}$. 
		In this figure, we let $\aut_{1} = \unfold(\imdp_{1})$ and $\aut_{2} = \unfold(\imdp_{2})$}
	\label{fig:schematic}
\end{figure}
To this aim and in order to pave the way for establishing the congruence result, we first prove two intermediate results stating that the folding and unfolding mappings preserve bisimilarity on the corresponding codomains.
\begin{lemma}
\label{lem:unfoldBisimPreservation}
	Given two \IMDPs $\imdp_{1}$ and $\imdp_{2}$, if $\imdp_{1} \bisim_{\un} \imdp_{2}$, then $\unfold(\imdp_{1}) \PAbisim \unfold(\imdp_{2})$.
\end{lemma}
\begin{myproof}
	Let $\relord$ be the probabilistic bisimulation justifying $\imdp_{1} \bisim_{\un} \imdp_{2}$; 
	we claim that $\relord$ is also a \PA probabilistic bisimulation for $\unfold(\imdp_{1})$ and $\unfold(\imdp_{2})$, that is, it justifies $\unfold(\imdp_{1}) \PAbisim \unfold(\imdp_{2})$.
	
	In the following we assume without loss of generality that $s_{1} \in \stateSet_{1}$ and $s_{2} \in \stateSet_{2}$;
	the other cases are similar.
	The fact that $\relord$ is an equivalence relation and that for each $(s_{1},s_{2}) \in \relord$, $\APLabelling_{1}(s_{1}) = \APLabelling_{2}(s_{2})$ follow directly by definition of $\bisim_{\un}$.
	Let $(s_{1}, \sd_{1}) \in \transitionRelation_{1}$:
	by definition of $\unfold$, it follows that $\sd_{1} \in \extreme{\uncertainty{s_{1}}{a_{1}}}$ for some $a_{1} \in \stateActionSet{s_{1}}$, thus in particular $\sd_{1} \in \uncertainty{s_{1}}{a_{1}}$, hence $\sd_{1} \in \convexhull{\bigcup_{a \in \stateActionSet{s_{1}}} \uncertainty{s_{1}}{a}}$.
	By hypothesis, we have that there exists $\sd_{2} \in \convexhull{\bigcup_{a_{2} \in \stateActionSet{s_{2}}} \uncertainty{s_{2}}{a_{2}}}$ such that $\sd_{1} \liftrel \sd_{2}$.
	Since $\sd_{2} \in \convexhull{\cup_{a_{2} \in \stateActionSet{s_{2}}} \uncertainty{s_{2}}{a_{2}}}$, it follows that there exist a multiset of real values $\setcond{p_{a_{2}} \in \posreals}{a_{2} \in \stateActionSet{s_{2}}}$ and a multiset of distributions $\setcond{\sd_{a_{2}} \in \uncertainty{s_{2}}{a_{2}}}{a_{2} \in \stateActionSet{s_{2}}}$ such that $\sum_{a_{2} \in \stateActionSet{s_{2}}} p_{a_{2}} = 1$ and $\sd_{2} = \sum_{a_{2} \in \stateActionSet{s_{2}}} p_{a_{2}} \cdot \sd_{a_{2}}$.
	For each $a_{2} \in \stateActionSet{s_{2}}$, since $\sd_{a_{2}} \in \uncertainty{s_{2}}{a_{2}}$, it follows that there exist a finite set of indexes $I_{a_{2}}$, a multiset of real values $\setcond{p_{a_{2},i} \in \posreals}{i \in I_{a_{2}}}$ and a multiset of distributions $\setcond{\sd_{a_{2},i} \in \extreme{\uncertainty{s_{2}}{a_{2}}}}{i \in I_{a_{2}}}$ such that $\sum_{i \in I_{a_{2}}} p_{a_{2},i} = 1$ and $\sd_{a_{2}} = \sum_{i \in I_{a_{2}}} p_{a_{2},i} \cdot \sd_{a_{2},i}$.
	This means that $\sd_{2} = \sum_{a_{2} \in \stateActionSet{s_{2}}} p_{a_{2}} \cdot \sum_{i \in I_{a_{2}}} p_{a_{2},i} \cdot \sd_{a_{2},i} = \sum_{a_{2} \in \stateActionSet{s_{2}}} \sum_{i \in I_{a_{2}}} p_{a_{2}} \cdot p_{a,i} \cdot \sd_{a_{2},i}$.
	Since for each $a_{2} \in \stateActionSet{s_{2}}$ and $i \in I_{a_{2}}$ we have that $\sd_{a_{2},i} \in \extreme{\uncertainty{s_{2}}{a_{2}}}$, it follows that $(s_{2}, \sd_{a_{2},i}) \in \transitionRelation_{2}$, thus we have the combined transition $\strongCombinedTransition{s_{2}}{\sd_{2}}$ obtained by taking as set of indexes $I = \setcond{(a_{2},i)}{a_{2} \in \stateActionSet{s_{2}}, i \in I_{a_{2}}}$, as multiset of real values $\setcond{q_{a_{2},i} \in \posreals}{(a_{2},i) \in I, q_{a_{2},i} = p_{a_{2}} \cdot p_{a_{2},i}}$, and as multiset of transitions $\setcond{(s_{2}, \sd_{a_{2},i}) \in \transitionRelation_{2}}{(a_{2},i) \in I}$:
	in fact, it is immediate to see that 
	\begin{align*}
	\sum_{(a_{2},i) \in I} q_{a_{2},i} 
		& = \sum_{(a_{2},i) \in I} p_{a_{2}} \cdot p_{a_{2},i} 
		  = \sum_{a_{2} \in \stateActionSet{s_{2}}} \sum_{i \in I_{a_{2}}} p_{a_{2}} \cdot p_{a_{2},i} \\
		& = \sum_{a_{2} \in \stateActionSet{s_{2}}} p_{a_{2}} \cdot \sum_{i \in I_{a_{2}}} p_{a_{2},i} 
		  = \sum_{a_{2} \in \stateActionSet{s_{2}}} p_{a_{2}} \cdot 1 
		  = 1
	\end{align*}
	and that 
	\begin{align*}
	& \hphantom{{}={}}\sum_{(a_{2},i) \in I} q_{a_{2},i} \cdot \sd_{a_{2},i} 
		  = \sum_{(a_{2},i) \in I} p_{a_{2}} \cdot p_{a_{2},i} \cdot \sd_{a_{2},i} \\
		& = \sum_{a_{2} \in \stateActionSet{s_{2}}} \sum_{i \in I_{a_{2}}} p_{a_{2}} \cdot p_{a_{2},i} \cdot \sd_{a_{2},i} 
		  = \sum_{a_{2} \in \stateActionSet{s_{2}}} p_{a_{2}} \cdot \sum_{i \in I_{a_{2}}} p_{a_{2},i} \cdot \sd_{a_{2},i} \\
		& = \sum_{a_{2} \in \stateActionSet{s_{2}}} p_{a_{2}} \cdot \sd_{a_{2}} 
		  = \sd_{2}\text{.}
	\end{align*}
	Moreover, by hypothesis, we have $\sd_{1} \liftrel \sd_{2}$, as required.
\end{myproof}
Likewise computation of probabilistic bisimulation for \IMDPs, we use the standard \emph{partition refinement} approach as a ground procedure to compute $\PAbisim$ for \PAs. 
Still the core part of the approach is to decide bisimilarity of a pair of states. 
For each state in the given \PA, we construct a convex hull polytope which encodes all possible behaviors that can be taken by a scheduler. 
Hence, for a given pair of states, we show that verifying if two states are bisimilar can be reduced to comparison of their corresponding convex polytopes with respect to set inclusion. 
Strictly speaking, for an equivalence relation $\relord$ on $\stateSet$ and $s \in \stateSet$, we denote by $\PAparcombuncertainty{s}$ the polytope of feasible successor distributions over \emph{equivalence classes} of $\relord$ with respect to taking a transition in the state $s$. 
Formally, 
\[
	\PAparcombuncertainty{s} = \convexhull{\setcond{\distclass{\sd}}{(s, \sd) \in \transitionRelation}}\text{,}
\]
where, for a given $\sd \in \Disc{\stateSet}$, $\distclass{\sd} \in \Disc{\partitionset{\stateSet}}$ is the probability distribution such that for each $\eqclass \in \partitionset{\stateSet}$, it is $\distclass{\sd}(\eqclass) = \sum_{s' \in \eqclass} \sd(s')$.
\begin{lemma}[cf.~{\cite[Thm.~1]{CS02}}]
\label{lem:PAbisimilarStatesSamePolytopes}
	Given a \PA $\aut$, there exists an equivalence relation $\relord$ on $\stateSet$ such that for each pair states $s, t \in \stateSet$, it holds that $s \PAbisim t$ if and only if $s \rel t$, $\APLabelling(s) = \APLabelling(t)$, and $\PAparcombuncertainty{s} = \PAparcombuncertainty{t}$.
\end{lemma}
To simplify the presentation of the proof, we first introduce some notation. 
Given an equivalence relation $\relord$ on $\stateSet$, for each distribution $\sd \in \Disc{\stateSet}$, let $\bar{\sd} \in \Disc{\partitionset{\stateSet}}$ denote the corresponding distribution $\bar{\sd} = \distclass{\sd}$, i.e., $\bar{\sd}$ is such that $\bar{\sd}(\eqclass) = \sum_{s' \in \eqclass} \sd(s')$ for each $\eqclass \in \partitionset{\stateSet}$.
\begin{myproof}
	We show the two implications separately.
	For the implication from left to right, suppose that $s \PAbisim t$;
	this implies that there exists a probabilistic bisimulation $\relord$ such that $s \rel t$ and $\APLabelling(s) = \APLabelling(t)$.
	We want to show that $\PAparcombuncertainty{s} = \PAparcombuncertainty{t}$ holds.
	To this aim, let $\eta \in \PAparcombuncertainty{s}$. 
	By definition of $\PAparcombuncertainty{s}$, it follows that there exist a finite set of indexes $I_{\eta}$, a multiset of real values $\setcond{p_{\eta, i}}{i \in I_{\eta}}$ and a multiset of distributions $\setcond{\eta_{i} \in \PAparcombuncertainty{s}}{i \in I_{\eta}, \exists (s,\sd_{s,i}) \in \transitionRelation: \eta_{i} = \distclass{\sd_{s,i}}}$ such that $\sum_{i \in I_{\eta}} p_{\eta,i} = 1$ and $\sum_{i \in I_{\eta}} p_{\eta,i} \cdot \eta_{i} = \eta$.
	Since $s \rel t$ and $\relord$ is a probabilistic bisimulation, it follows that for each $i \in I_{\eta}$ there exists a combined transition $\strongCombinedTransition{t}{\sd_{t}}$ such that $\sd_{s,i} \liftrel \sd_{t}$.
	By definition of combined transition, it follows that there exist a finite set of indexes $I_{t}$, a set of transitions $\setcond{(t, \sd_{t,i}) \in \transitionRelation}{i \in I_{t}}$ and a multiset of real values $\setcond{p_{t,i} \in \posreals}{i \in I_{t}}$ such that $\sum_{i \in I_{t}} p_{t,i} = 1$ and $\sd_{t} = \sum_{i \in I_{t}} p_{t,i} \cdot \sd_{t,i}$.
	This implies that for each $i \in I_{t}$, $\bar{\sd}_{t,i} \in \PAparcombuncertainty{t}$.
	Moreover, since by definition of lifting we have that for each $\eqclass \in \partitionset{\stateSet}$, $\sd_{s}(\eqclass) = \sd_{t}(\eqclass)$, it follows immediately that $\bar{\sd}_{t} = \bar{\sd}_{s}$, thus we have that $\eta = \bar{\sd}_{s} = \bar{\sd}_{t} \in \PAparcombuncertainty{t}$, hence $\PAparcombuncertainty{s} \subseteq \PAparcombuncertainty{t}$.
	By swapping the roles of $s$ and $t$, we can show in the same way that $\PAparcombuncertainty{t} \subseteq \PAparcombuncertainty{s}$, hence $\PAparcombuncertainty{s} = \PAparcombuncertainty{t}$ as required.
	
	For the implication from right to left, fix an equivalence relation $\relord$ on $\stateSet$ such that for each $(s,t) \in \relord$ it holds that $\APLabelling(s) = \APLabelling(t)$ and $\PAparcombuncertainty{s} = \PAparcombuncertainty{t}$;
	we want to show that $\relord$ is a probabilistic bisimulation, i.e., whenever $s \rel t$ and $\strongTransition{s}{\sd_{s}}$ then there exists $\strongCombinedTransition{t}{\sd_{t}}$ such that $\sd_{s} \liftrel \sd_{t}$.
	Let $(s,t) \in \relord$;
	if $\PAparcombuncertainty{s} = \emptyset$, then the step condition of the probabilistic bisimulation is trivially verified since there is no transition $\strongTransition{s}{\sd_{s}}$ from $s$ that needs to be matched by $t$.
	Suppose now that $\PAparcombuncertainty{s} \neq \emptyset$ and consider a transition $\strongTransition{s}{\sd_{s}}$ so that $\bar{\sd}_{s} \in \PAparcombuncertainty{s}$.
	By hypothesis, $\bar{\sd}_{s} \in \PAparcombuncertainty{s} = \PAparcombuncertainty{t}$, thus there exist a finite set of indexes $I$, a multiset of distributions $\setcond{\sd_{i} \in \PAparcombuncertainty{t}}{i \in I}$ and a multiset of real values $\setcond{p_{i} \in \posreals}{i \in I}$ such that $\sum_{i \in I} p_{i} = 1$ and $\sum_{i \in I} p_{i} \cdot \sd_{i} = \bar{\sd}_{s}$.
	This implies, for each $i \in I$, that there exist a finite set of indexes $J_{i}$, a multiset of real values $\setcond{p_{i,j} \in \posreals}{j \in I_{i}}$, and a multiset of distributions $\setcond{\sd_{i,j} \in \PAparcombuncertainty{t}}{j \in J_{i}}$ such that $\sum_{j \in J_{i}} p_{i,j} = 1$, $\sum_{j \in J_{i}} p_{i,j} \cdot \sd_{i,j} = \sd_{i}$, and for each $j \in J_{i}$, $\sd_{i,j} = \bar{\sd}_{t,i,j}$ where $(t, \sd_{t,i,j}) \in \transitionRelation$.
	Consider now the combined transition $\strongCombinedTransition{t}{\sd_{t}}$ obtained by taking as set of indexes $J = \setcond{(i,j)}{i \in I, j \in J_{i}}$, as multiset of real values $\setcond{q_{i,j} \in \posreals}{(i,j) \in J, q_{i,j} = p_{i} \cdot p_{i,j}}$, and as set of transitions $\setcond{(t, \sd_{t,i,j}) \in \transitionRelation}{(i,j) \in J}$:
	we have that 
	\begin{align*}
		\sum_{(i,j) \in J} q_{i,j}
			& = \sum_{(i,j) \in J} p_{i} \cdot p_{i,j} 
			  = \sum_{i \in I} \sum_{j \in J_{i}} p_{i} \cdot p_{i,j} \\
			& = \sum_{i \in I} p_{i} \cdot \sum_{j \in J_{i}} p_{i,j} 
			  = \sum_{i \in I} p_{i} \cdot 1 
			  = 1 
	\end{align*}
	and that 
	\begin{align*}
		\sd_{t} 
			& = \sum_{(i,j) \in I} q_{i,j} \cdot \sd_{t,i,j} 
			  = \sum_{(i,j) \in I} p_{i} \cdot p_{i,j} \cdot \sd_{t,i,j} \\
			& = \sum_{i \in I} \sum_{j \in J_{i}} p_{i} \cdot p_{i,j} \cdot \sd_{t,i,j} 
			  = \sum_{i \in I} p_{i} \cdot \sum_{j \in J_{i}} p_{i,j} \cdot \sd_{t,i,j}\text{.}
	\end{align*}
	To complete the proof, we have to show that $\sd_{s} \liftrel \sd_{t}$, that is, for each $\eqclass \in \partitionset{\stateSet}$, $\sd_{s}(\eqclass) = \sd_{t}(\eqclass)$.
	Let $\eqclass \in \partitionset{\stateSet}$:
	we have that 
		\begin{align*}
		\sd_{t}(\eqclass)
		 & = \sum_{c \in \eqclass} \sd_{t}(c) 
		   = \sum_{c \in \eqclass} \sum_{i \in I} p_{i} \cdot \sum_{j \in J_{i}} p_{i,j} \cdot \sd_{t,i,j}(c) \\
		 & = \sum_{i \in I} p_{i} \cdot \sum_{j \in J_{i}} p_{i,j} \cdot \sum_{c \in \eqclass} \sd_{t,i,j}(c) 
		   = \sum_{i \in I} p_{i} \cdot \sum_{j \in J_{i}} p_{i,j} \cdot \bar{\sd}_{t,i,j}(\eqclass) \\
		 & = \sum_{i \in I} p_{i} \cdot \sum_{j \in J_{i}} p_{i,j} \cdot \sd_{i,j}(\eqclass) 
		   = \sum_{i \in I} p_{i} \cdot \sd_{i}(\eqclass) 
		   = \bar{\sd}_{s}(\eqclass) \\
		 & = \sum_{c \in \eqclass} \sd_{s}(c) 
	       = \sd_{s}(\eqclass)\text{,}
		\end{align*}
	as required.
\end{myproof}

\begin{lemma}
\label{lem:polytopeProjectEquivalence}
	Given a \PA $\aut$ and an equivalence relation $\relord$ on $\stateSet$, for $n = \setcardinality{\partitionset{\stateSet}}$, it holds that for each $(s, t) \in \relord$, if $\PAparcombuncertainty{s} = \PAparcombuncertainty{t}$ then $(\prod_{\eqclass \in \partitionset{\stateSet}} \vctproj[\eqclass]{\PAparcombuncertainty{s}}) \cap \Delta_{n} = (\prod_{\eqclass \in \partitionset{\stateSet}} \vctproj[\eqclass]{\PAparcombuncertainty{t}}) \cap \Delta_{n}$.
\end{lemma}
\begin{proof}
	The proof is trivial, since by $\PAparcombuncertainty{s} = \PAparcombuncertainty{t}$ it follows that for each $\eqclass \in \partitionset{\stateSet}$, $\vctproj[\eqclass]{\PAparcombuncertainty{s}} = \vctproj[\eqclass]{\PAparcombuncertainty{t}}$.
	This implies that $\prod_{\eqclass \in \partitionset{\stateSet}} \vctproj[\eqclass]{\PAparcombuncertainty{s}} = \prod_{\eqclass \in \partitionset{\stateSet}} \vctproj[\eqclass]{\PAparcombuncertainty{t}}$ thus $(\prod_{\eqclass \in \partitionset{\stateSet}} \vctproj[\eqclass]{\PAparcombuncertainty{s}}) \cap \Delta_{n} = (\prod_{\eqclass \in \partitionset{\stateSet}} \vctproj[\eqclass]{\PAparcombuncertainty{t}}) \cap \Delta_{n}$, as required.
\end{proof}

\begin{lemma}
\label{lem:foldBisimPreservation} 
	Given two \PAs $\aut_{1}$ and $\aut_{2}$, if $\aut_{1} \PAbisim \aut_{2}$ then $\fold(\aut_{1}) \bisim_{\un} \fold(\aut_{2})$.
\end{lemma}
\begin{myproof}
	Let $\relord$ be the equivalence relation justifying $\aut_{1} \PAbisim \aut_{2}$;
	we claim that $\relord$ is also an \IMDP probabilistic bisimulation for $\fold(\aut_{1})$ and $\fold(\aut_{2})$, that is, it justifies $\fold(\aut_{1}) \bisim_{\un} \fold(\aut_{2})$.
	
	In the following we assume without loss of generality that $s_{1} \in \stateSet_{1}$ and $s_{2} \in \stateSet_{2}$;
	the other cases are similar.
	The fact that $\relord$ is an equivalence relation and that for each $(s_{1},s_{2}) \in \relord$, $\APLabelling_{1}(s_{1}) = \APLabelling_{2}(s_{2})$ follow directly by definition of $\PAbisim$.
	Since $\aut_{1} \PAbisim \aut_{2}$, it follows from Lemma~\ref{lem:PAbisimilarStatesSamePolytopes} that $\PAparcombuncertainty{s_{1}} = \PAparcombuncertainty{s_{2}}$. 
	Additionally, it is not difficult to see that for $s_{j} \in \setnocond{s_{1}, s_{2}}$, $\parcombuncertainty{s_j} = \paruncertainty{s_j}{f} = (\prod_{i=1}^{n}\vctproj[i]{\PAparcombuncertainty{s_{j}})\cap \Delta_{n}} $ where $\Delta_{n} = \setcond{(x_1,\dotsc,x_{n}) \in \posreals^{n}}{\sum_{i=1}^{n} x_{i} = 1}$ is the standard probability simplex in $\reals^{n}$. 
	By Lemma~\ref{lem:polytopeProjectEquivalence}, this implies that $\paruncertainty{s_1}{f} = \paruncertainty{s_2}{f}$.  
	Consider now $\transition{s_{1}}{\sd_{1}}$ with $\sd_{1} \in\uncertainty{s_{1}}{f}$:
	this implies that $\distclass{\sd_{1}} \in \paruncertainty{s_{1}}{f}$
	since $\paruncertainty{s_{1}}{f} = \paruncertainty{s_{2}}{f}$, it follows that $\distclass{\sd_{1}} \in \paruncertainty{s_{2}}{f}$ as well, thus there exists $\sd_{2} \in \uncertainty{s_{2}}{f}$ such that $\distclass{\sd_{2}} = \distclass{\sd_{1}}$. 
	By definition of $\distclass{\functionDot}$, we have that for each $\eqclass \in \partitionset{\stateSet}$, $\sum_{s \in \eqclass} \sd_{i}(s) = \distclass{\sd_{i}}$ for $i \in \setnocond{1,2}$, thus $\distclass{\sd_{2}} = \distclass{\sd_{1}}$ implies that for each $\eqclass \in \partitionset{\stateSet}$, $\sum_{s \in \eqclass} \sd_{1}(s) = \sum_{s \in \eqclass} \sd_{2}(s)$, i.e., $\sd_{1} \liftrel \sd_{2}$.
	This means that we have found $\transition{s_{2}}{\sd_{2}}$ with $\sd_{1} \liftrel \sd_{2}$, as required.
\end{myproof}

By using Lemmas~\ref{lem:unfoldBisimPreservation} and~\ref{lem:foldBisimPreservation} and Proposition~\ref{pro:PAbisimSyncProduct}, we can now show that $\bisim_{\un}$ is preserved by the synchronous product operator introduced in Definition~\ref{def:IMDPparallelComposition}. 
\begin{theorem}
\label{thm:IMDPBisimSyncProdCongruence}
	Given three \IMDPs $\imdp_{1}$, $\imdp_{2}$, and $\imdp_{3}$, if $\imdp_{1} \bisim_{\un} \imdp_{2}$, then $\syncProd{\imdp_{1}}{\imdp_{3}} \bisim_{\un} \syncProd{\imdp_{2}}{\imdp_{3}}$. 
\end{theorem}
\begin{myproof}
	Assume that $\imdp_{1} \bisim_{\un} \imdp_{2}$. 
	By Lemma~\ref{lem:unfoldBisimPreservation}, it follows that $\unfold(\imdp_{1}) \PAbisim \unfold(\imdp_{2})$, thus, by Proposition~\ref{pro:PAbisimSyncProduct}, we have that $\syncProd{\unfold(\imdp_{1})}{\unfold(\imdp_{3})} \PAbisim \syncProd{\unfold(\imdp_{2})}{\unfold(\imdp_{3})}$.
	Lemma~\ref{lem:foldBisimPreservation} now implies that $\fold(\syncProd{\unfold(\imdp_{1})}{\unfold(\imdp_{3})}) \bisim_{\un} \fold(\syncProd{\unfold(\imdp_{2})}{\unfold(\imdp_{3})})$, that is, $\syncProd{\imdp_{1}}{\imdp_{3}} \bisim_{\un} \syncProd{\imdp_{2}}{\imdp_{3}}$, as required.
\end{myproof}

\section{Interleaved approach}
In the previous sections, we have considered the parallel composition via synchronous production, which is working by the definition of folding collapsing all labels to a single transition.
Here we consider the other extreme of the parallel composition: interleaving only.
\begin{definition}
\label{def:IMDPinterleavedParallelComposition}
	Given two \IMDPs $\imdp_{l}$ and $\imdp_{r}$, we define the \emph{interleaved composition} of $\imdp_{l}$ and $\imdp_{r}$, denoted by $\interleavedProd{\imdp_{l}}{\imdp_{r}}$, as the \IMDP $\imdp = (\stateSet, \startState, \actionSet, \APSet, \APLabelling, \intTransitionProbability)$ where 
		$\stateSet = \stateSet_{l} \times \stateSet_{r}$;
		$\startState = (\startState_{l},\startState_{r})$;
		$\actionSet = (\actionSet_{l} \times \setnocond{l}) \cup (\actionSet_{r} \times \setnocond{r})$; 
		$\APSet = \APSet_{l} \cup \APSet_{r}$;
		for each $(s_{l},s_{r}) \in \stateSet$, $\APLabelling(s_{l},s_{r}) = \APLabelling_{l}(s_{l}) \cup \APLabelling_{r}(s_{r})$;
		and
		\[
			\intTransitionProbability((s_{l},s_{r}), (a,i), (t_{l},t_{r})) =
				\begin{cases}
					\intTransitionProbability_{l}(s_{l}, a, t_{l}) & \text{if $i = l$ and $t_{r} = s_{r}$,} \\
					\intTransitionProbability_{r}(s_{r}, a, t_{r}) & \text{if $i = r$ and $t_{l} = s_{l}$,} \\
					\interval{0}{0} & \text{otherwise.}
				\end{cases}
		\]
\end{definition}

\begin{theorem}
\label{thm:IMDPBisimInterleavedProdCongruence}
	Given three \IMDPs $\imdp_{1}$, $\imdp_{2}$, and $\imdp_{3}$, if $\imdp_{1} \bisim_{\un} \imdp_{2}$, then $\interleavedProd{\imdp_{1}}{\imdp_{3}} \bisim_{\un} \interleavedProd{\imdp_{2}}{\imdp_{3}}$.
\end{theorem}
\begin{myproof}
	Let $\relord$ be the probabilistic bisimulation justifying $\imdp_{1} \bisim_{\un} \imdp_{2}$ and define $\relord' = \relord \times \idrelord_{\stateSet_{3}}$;
	we claim that $\relord'$ is a probabilistic bisimulation between $\interleavedProd{\imdp_{1}}{\imdp_{3}}$ and $\interleavedProd{\imdp_{2}}{\imdp_{3}}$.
	The fact that $\relord'$ is an equivalence relation follows trivially by its definition and the fact that $\relord$ is an equivalence relation.
	The fact that $((\startState_{1},\startState_{3}), (\startState_{2},\startState_{3}))$ follows immediately by the hypothesis that $(\startState_{1}, \startState_{2}) \in \relord$ and $(\startState_{3}, \startState_{3}) \in \idrelord_{\stateSet_{3}}$.

	Let $((s_{1},s_{3}),(s_{2},s_{3})) \in \relord'$.
	Assume, without loss of generality, that $s_{1} \in \stateSet_{1}$ and $s_{2} \in \stateSet_{2}$;
	the other cases are similar.
	The fact that $\APLabelling_{1,3}(s_{1},s_{3}) = \APLabelling_{2,3}(s_{2},s_{3})$ is straightforward, since by definition of interleaved composition and the hypothesis that $s_{1} \rel s_{2}$, we have that $\APLabelling_{1,3}(s_{1},s_{3}) = \APLabelling_{1}(s_{1}) \cup \APLabelling_{3}(s_{3}) = \APLabelling_{2}(s_{2}) \cup \APLabelling_{3}(s_{3}) = \APLabelling_{2,3}(s_{2},s_{3})$, as required.

	Consider now a transition $\transition{(s_{1},s_{3})}{\sd_{1,3}}$.
	By definition, we have that $\sd_{1,3} \in \convexComb{\bigcup_{(a,i) \in \stateActionSet{s_{1},s_{3}}} \uncertainty{(s_{1},s_{3})}{(a,i)}}$.
	This implies that there exist a multiset of distributions $\setcond{\sd_{a,i} \in \uncertainty{(s_{1},s_{3})}{(a,i)}}{(a,i) \in \stateActionSet{s_{1},s_{3}}}$ and a multiset of real values $\setcond{p_{a,i} \in \posreals}{(a,i) \in \stateActionSet{s_{1},s_{3}}}$ such that $\sum_{(a,i) \in \stateActionSet{s_{1},s_{3}}} p_{a,i} = 1$ and $\sum_{(a,i) \in \stateActionSet{s_{1},s_{3}}} p_{a,i} \cdot \sd_{a,i} = \sd_{1,3}$.
	Consider an action $(a,i) \in \stateActionSet{s_{1},s_{3}}$:
	by definition of interleaved composition, it is either of the form $(a,l) \in \actionSet_{l} \times \setnocond{l}$, or of the form $(a,r) \in \actionSet_{r} \times \setnocond{r}$.
	Consider the two cases separately:
	\begin{description}
	\item[Case $(a,i) \in \actionSet_{l} \times \setnocond{l}$:]
		this means that $\sd_{a,i}$ is actually the distribution $\sd_{a,i} = \sd_{a} \times \dirac{s_{3}}$ where $\sd_{a} \in \uncertainty{s_{1}}{a}$ is such that for each $s'_{1} \in \stateSet_{1}$, $\sd_{a}(s'_{1}) = \sd_{a,i}(s'_{1}, s_{3})$, thus $\transition{s_{1}}{\sd_{a}}$.
		Since by hypothesis $(s_{1}, s_{2}) \in \relord$ and $\relord$ is a probabilistic bisimulation, there exists $\sd_{a,2} \in \convexComb{\bigcup_{b \in \stateActionSet{s_{2}}} \uncertainty{s_{2}}{b}}$ such that $\sd_{a} \liftrel \sd_{a,2}$.
		This implies that there exist a multiset of distributions $\setcond{\sd_{a,2,b} \in \uncertainty{s_{2}}{b}}{b \in \stateActionSet{s_{2}}}$ and a multiset of real values $\setcond{p_{a,2,b}}{b \in \stateActionSet{s_{2}}}$ such that $\sum_{b \in \stateActionSet{s_{2}}} p_{a,2,b} = 1$, $\sum_{b \in \stateActionSet{s_{2}}} p_{a,2,b} \cdot \sd_{a,2,b} = \sd_{a,2}$, and $\sd_{a} \times \dirac{s_{3}} \liftrel[\relord'] \sd_{a,2} \times \dirac{s_{3}}$.
		This means that for each $b \in \stateActionSet{s_{2}}$, we have that $\sd_{a,2,b} \times \dirac{s_{3}} \in \uncertainty{(s_{2},s_{3})}{(b,l)}$, thus by taking $\sd_{a,2,3} = \sum_{b \in \stateActionSet{s_{2}}} p_{a,2,b} \cdot (\sd_{a,2,b} \times \dirac{s_{3}})$ we have that $\transition{(s_{2},s_{3})}{\sd_{a,2,3}}$ and $\sd_{a,i} \liftrel[\relord'] \sd_{a,2,3}$.

	\item[Case $(a,i) \in \actionSet_{r} \times \setnocond{r}$:]
		this means that $\sd_{a,i}$ is actually the distribution $\sd_{a,i} = \dirac{s_{1}} \times \sd_{a}$ where $\sd_{a} \in \uncertainty{s_{3}}{a}$ is such that for each $s'_{3} \in \stateSet_{3}$, $\sd_{a}(s'_{3}) = \sd_{a,i}(s_{1}, s'_{3})$, thus $\transition{s_{3}}{\sd_{a}}$.
		This implies trivially that $\transition{(s_{2},s_{3})}{\sd_{a,2,3}}$ where $\sd_{a,2,3} = \dirac{s_{2}} \times \sd_{a}$ and $\sd_{a,i} \liftrel[\relord'] \sd_{a,2,3}$.
	\end{description}
	From the analysis of the two cases, we have that for each $(a,i) \in \stateActionSet{s_{1},s_{3}}$, there exists a transition $\transition{(s_{2},s_{3})}{\sd_{a,2,3}}$ such that $\sd_{a,i} \liftrel[\relord'] \sd_{a,2,3}$.
	This implies that $\sd_{2,3} = \sum_{(a,i) \in \stateActionSet{s_{1},s_{3}}} p_{a,i} \cdot \sd_{a,2,3} \in \convexComb{\bigcup_{(b,j) \in \stateActionSet{s_{2},s_{3}}} \uncertainty{(s_{2},s_{3})}{(b,j)}}$ and $\sd_{1,3} \liftrel[\relord'] \sd_{2,3}$, as required.
\end{myproof}

\section{Concluding Remarks} 
\label{sec:conclusion}
In this paper, we have studied the probabilistic bisimulation problem for interval \MDPs in order to 
speed up the run time of model checking algorithms that often suffer from the state space explosion.
Interval MDPs include two sources of nondeterminism for which we have considered the cooperative resolution in a dynamic setting.
We have revised and extended the compositionality reasoning in~\cite{HashemiHSSTW16} by further exploration on the possibility of defining 
the parallel operator for \IMDP models which preserve our notion of probabilistic bisimulation.

\paragraph{Acknowledgments}   
This work is supported 
by the EU 7th Framework Programme under grant agreements 295261 (MEALS) and 318490 (SENSATION), 
by the DFG as part of SFB/TR 14 AVACS, 
by the ERC Advanced Investigators Grant 695614 (POWVER),
by the CAS/SAFEA International Partnership Program for Creative Research Teams, 
by the National Natural Science Foundation of China (Grants No.\ 61472473, 61532019, 61550110249, 61550110506),
by the Chinese Academy of Sciences Fellowship for International Young Scientists, 
and 
by the CDZ project CAP (GZ 1023).

\bibliography{biblio}
\bibliographystyle{abbrv}
\end{document}